\documentclass[11pt]{article}
\usepackage[utf8]{inputenc}
\usepackage[margin=1in]{geometry}
\usepackage{amsmath,amssymb,amsthm}
\usepackage{graphicx}
\usepackage{hyperref}
\usepackage{enumitem}
\usepackage{tikz}
\usepackage{float}
\usepackage{caption}
\usepackage{subcaption}

\newtheorem{theorem}{Theorem}
\newtheorem{lemma}[theorem]{Lemma}
\newtheorem{corollary}[theorem]{Corollary}
\theoremstyle{definition}
\newtheorem{definition}{Definition}
\newtheorem{example}{Example}

\title{\textbf{Context Lake: A System Class Defined by Decision Coherence}\\
\large Correctness for Collective AI Systems\\
\vspace{0.2cm}
\large \textit{Position Paper}}

\author{
Xiaowei Jiang\\
\texttt{jiangxw@apache.org}
\thanks{The author works at Tacnode, which is building a Context Lake implementation.}
}

\date{January 2026\\
\small Preprint}

\begin{document}

\maketitle

\begin{abstract}
AI agents are increasingly the primary consumers of data, operating continuously to make concurrent, irreversible decisions. Traditional data systems designed for human analysis cycles become correctness bottlenecks under this operating regime. When multiple agents operate over shared resources, their actions interact before reconciliation is possible. Correctness guarantees that apply after the decision window therefore fail to prevent conflicts.

We introduce the \textbf{Decision Coherence Law}: \textit{for agents that take irreversible actions whose effects interact, correctness requires that interacting decisions be evaluated against a coherent representation of reality at the moment they are made.} We show that no existing system class satisfies this requirement and prove through the \textbf{Composition Impossibility Theorem} that independently advancing systems cannot be composed to provide Decision Coherence while preserving their native system classes.

From this impossibility result, we derive \textbf{Context Lake} as a necessary system class with three requirements: (1)~semantic operations as native capabilities, (2)~transactional consistency over all decision-relevant state, and (3)~operational envelopes bounding staleness and degradation under load. We formalize the architectural invariants, enforcement boundaries, and admissibility conditions required for correctness in collective agent systems.

This position paper establishes the theoretical foundation for Context Lakes, identifies why existing architectures fail, and specifies what systems must guarantee for AI agents to operate constructively at scale.
\end{abstract}


\noindent\textbf{ACM CCS Concepts:} \textbf{Information systems} $\rightarrow$ Database management system engines; Data management systems; \textbf{Computing methodologies} $\rightarrow$ Multi-agent systems; \textbf{Theory of computation} $\rightarrow$ Distributed algorithms

\section{Introduction: The Foundational Shift}

\subsection{The Primary Consumer Has Changed}

Traditional databases, data warehouses, and analytics platforms were designed for human analysis cycles. While some systems have served automated processes, the dominant design center has been retrospective analysis rather than continuous decision making under concurrency. Humans analyze data. Agents act on context. That difference fundamentally changes the requirements.

AI agents operate on millisecond decision loops: Observe. Decide. Act. Repeat. Continuously, many times per second. When they require context to make a decision, ``an hour old'' is ancient history and even ``a minute old'' is outdated. Agents require data as it exists in that instant, not as it was at the time of the last batch update.

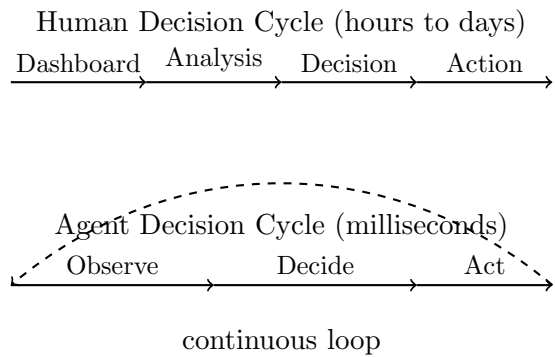
\begin{figure}[h]
\centering
\begin{tikzpicture}[scale=0.9]
\draw[->, thick] (0,2) -- (2,2) node[midway, above] {\small Dashboard};
\draw[->, thick] (2,2) -- (4,2) node[midway, above] {\small Analysis};
\draw[->, thick] (4,2) -- (6,2) node[midway, above] {\small Decision};
\draw[->, thick] (6,2) -- (8,2) node[midway, above] {\small Action};
\node[above] at (4,2.5) {Human Decision Cycle (hours to days)};

\draw[->, thick] (0,-1) -- (3,-1) node[midway, above] {\small Observe};
\draw[->, thick] (3,-1) -- (6,-1) node[midway, above] {\small Decide};
\draw[->, thick] (6,-1) -- (8,-1) node[midway, above] {\small Act};
\draw[->, thick, dashed] (8,-1) to[bend right=40] (0,-1);
\node[above] at (4,-0.5) {Agent Decision Cycle (milliseconds)};
\node[below] at (4,-1.5) {continuous loop};
\end{tikzpicture}
\caption{Human vs. Agent decision cycles. Human decisions occur in discrete cycles with tolerance for staleness. Agent decisions occur continuously with interactions before reconciliation is possible.}\label{fig:decision-cycles}
\end{figure}

\textbf{Scope.} This paper addresses systems where AI agents---not humans---are the primary consumers of data, and where agents operate \textit{constructively} such that intelligence compounds rather than remains isolated.

By ``agents'' we mean AI systems with two defining characteristics:
\begin{enumerate}[leftmargin=*]
\item \textbf{Semantic understanding:} Interpreting meaning directly from unstructured content---text, images, audio, video---at decision time.
\item \textbf{Continuous operation:} Acting in perpetual motion without batch boundaries or quiescent periods.
\end{enumerate}

For single-agent systems or independent agents with isolated contexts, the requirements established here do not apply.

\subsection{The Memory Bottleneck}

Changing the primary consumer from humans to agents does more than accelerate decision cycles. It changes where understanding must reside.

Agents act continuously and concurrently. To act correctly, an agent must reason over far more context than can be held within a single execution: recent events, evolving state, shared commitments, historical patterns, and interpretations produced by other agents. This context is not static. It evolves as agents act.

Under these conditions, the limiting factor shifts from the ability to infer or decide to the ability to retain and carry forward context as agents continue to operate.

\textbf{Model memory limitations.} Models retain information in two ways: through their parameters and through execution-local context. Parametric memory encodes statistical regularities but cannot store explicit facts about the current world, cannot be updated transactionally, and cannot be corrected atomically. Execution-local contextual memory is ephemeral and instance-bound---it exists only within that execution and is invisible to other concurrent agents.

This constraint is not unique to artificial systems. For most of human history, knowledge accumulated by one individual vanished with them unless externalized. Intelligence did not compound reliably until memory moved outside the human brain through writing, libraries, and the internet.

AI systems now face an analogous inflection point. If each agent keeps its own memory, we get millions of isolated intelligences. If memory becomes shared, breakthroughs propagate instantly across all agents at machine speed. \textbf{To compound, memory must become infrastructure.}

\subsection{What Shared Memory Enables}

Consider a warehouse fulfillment system with three autonomous agents operating concurrently: \textit{Inventory Agent} (monitors stock), \textit{Shipping Agent} (processes orders), and \textit{Restocking Agent} (receives deliveries). Each makes dozens of decisions per minute.

\begin{example}[Constructive Operation]
Three agents operate concurrently over a shared, transactional inventory state.
\begin{itemize}[leftmargin=*,noitemsep]
\item \texttt{14:23:18.300} — \textit{Restocking Agent} processes a returned unit that had been counted as available, discovers it is defective, and submits a correction to inventory.
\item \texttt{14:23:18.310} — \textit{Inventory Agent} applies the correction transactionally, updating available inventory: 2 units $\rightarrow$ 1 unit.
\item \texttt{14:23:18.350} — \textit{Shipping Agent}, processing an order that requires 2 units, sees: 1 unit available.
\item \texttt{14:23:18.400} — \textit{Shipping Agent} recognizes insufficient inventory, escalates for split shipment approval.
\end{itemize}
\end{example}

This occurs in 100~milliseconds. The system avoided an invalid inventory commitment because the agents operated over the same coherent representation of reality. Without shared context, the Shipping Agent would have committed the order against stale data.

This example is intentionally simple—all decision-relevant state is deterministic, transactional, and already shared. In real systems, the most consequential failures occur not when facts are missing, but when correct knowledge exists and cannot participate in a decision.

\subsection{When Knowledge Exists but Is Trapped in Agent Silos}

Consider an online commerce system at checkout. The system is intentionally organized with operational decisions in a database and behavioral analysis in a lakehouse. At decision time:

The \textit{checkout agent} evaluates the order using operational database state. Every signal appears normal: clean transaction history, amount within range, shipping address verified, device and location match recent activity. The agent authorizes the transaction.

Simultaneously, a \textit{behavior agent} processes clickstream data in the lakehouse. It derives a weak but meaningful pattern: direct arrival on checkout URL, purchase without browsing. By itself, this is common and non-definitive. But it is a known precursor in account-takeover scenarios when combined with an otherwise normal-looking purchase. The behavior agent records this interpretation in the lakehouse.

The checkout agent never observes this knowledge—not because it ignores behavioral data, nor because the behavior agent failed to compute it, but because the interpretation exists in a system the checkout agent cannot consult within its decision window.

The laptop ships. Thirty-six hours later, the charge is disputed. The account was compromised. The attacker kept the transaction within normal bounds, relying on the fact that the only early warning signal existed as behavioral knowledge trapped outside the checkout agent's decision context.

This failure was not caused by missing data, delayed processing, or a bad model. Knowledge was formed correctly and on time. The failure was an \textbf{agent silo}.

\subsection{The Structural Question}

When systems fail to meet agent requirements, the default response is incremental: faster pipelines, better indexes, smarter caching. But speed alone cannot close this gap. The problem is structural, not operational.

This paper establishes what agents require and why existing system classes fail to provide it. But understanding requirements is not the same as understanding purpose.

When agents act in isolation, their behavior does not compound. When multiple agents operate over shared resources or state, their decisions either reinforce one another or interfere. Reinforcement is not automatic: it occurs only when information that influences one agent's decision is visible to others whose decisions interact. Absent this condition, concurrent operation amplifies error rather than intelligence.

Under this operating regime, a fundamental question emerges:

\begin{center}
\textit{What constraint must hold for agent action to be constructive—when decisions overlap in time, interact through shared state, and produce irreversible effects?}
\end{center}

This is not a question about performance optimization or system composition. It is a question about the minimal requirement for agent operations to produce compounding intelligence rather than compounding errors.

\section{What AI Agents Require}

Agents do not merely analyze data. They act continuously, concurrently, and irreversibly in the real world. This operating regime imposes requirements on context that existing data systems were not designed to satisfy.

\subsection{Agents Require Semantic Operations}

Earlier AI systems handled meaning in bounded ways. Interpretation occurred offline during training, or was delegated to specialized subsystems operating over fixed semantic spaces. At decision time, decisions were evaluated against pre-computed features whose interpretation was stable and externally defined.

Large language models (LLMs) change this. Agents can now interpret raw inputs—text, images, audio, video—on demand, as decisions are being made. They extract intent, resolve entities, assess similarity, and classify activity based on meaning.

Agent decisions often depend on semantic operations such as:
\begin{itemize}[leftmargin=*,noitemsep]
\item interpreting free-form input to determine intent or risk
\item retrieving situations similar in meaning rather than matching identifiers
\item classifying activity into pattern families
\end{itemize}

For agents, semantic interpretation is not preprocessing or analysis—it is part of the decision itself, and deferring it changes which decision is made.

\subsection{Agents Need Derived Context}

Raw data records what happened. Decisions also depend on \textbf{derived context}: representations produced by interpreting and consolidating raw observations over time. Examples include rolling aggregates, deviations from baseline, correlated activity indicators, and similarity structures.

As agents operate continuously and decisions overlap in time, derived context accumulates and evolves alongside raw data. Both raw and derived context evolve continuously under concurrency, with no quiescent point between updates and decisions.

\subsection{Agents Require Many Retrieval Patterns}

A single agent decision typically composes multiple dependent retrievals over raw and derived context:
\begin{itemize}[leftmargin=*,noitemsep]
\item point lookups over current state
\item range scans over recent history
\item filters or aggregations over dynamically defined cohorts
\item secondary-index access over non-key attributes
\item similarity retrieval over high-dimensional representations
\item semantic retrieval based on interpreted intent and conceptual relevance
\end{itemize}

These retrievals may be causally ordered, with latencies that compose. All must complete within a bounded decision window. If any retrieval observes stale or delayed inputs relative to others, the decision is evaluated against a combination of facts that did not exist together at any single moment.

\subsection{Existing System Classes Fragment Decision-Time Context}

Current system classes divide responsibility along clean architectural boundaries:

\begin{itemize}[leftmargin=*]
\item \textbf{OLTP systems} efficiently mutate primary records but are not designed to continuously maintain large volumes of derived context or support complex, multi-pattern retrieval under high concurrency.
\item \textbf{Analytical systems} compute rich derived results at scale but operate over historical snapshots, trading freshness and concurrency for scale and query flexibility.
\item \textbf{Search engines} support efficient keyword retrieval but do not continuously maintain derived context, and make updates visible according to indexing policies rather than mutation timing.
\item \textbf{Vector databases} optimize similarity search over vector representations but do not own the additional context against which decisions are evaluated.
\end{itemize}

Each system is effective within its intended scope. The limitation arises because agent decisions must combine raw and derived context, using multiple retrieval patterns, over continuously evolving state, under concurrency, at decision time.

This requirement is sometimes conflated with Hybrid Transactional/Analytical Processing (HTAP) systems. However, HTAP has never been defined as a system class with formal semantics or invariants. It is a descriptive label for a workload pattern rather than an abstraction that defines how decisions are evaluated under concurrency.

No existing system class was designed to enforce these guarantees jointly.

\section{Decision Coherence}

\subsection{The Decision Coherence Law}

Under constructive operation, intelligence compounds: one agent's work informs others' decisions, patterns propagate across agents, and understanding accumulates in shared memory rather than remaining isolated.

But compounding is only possible when agents act on the same understanding of reality. If Agent A updates its understanding but Agent B cannot see that update---if they operate from incompatible representations---then Agent B cannot benefit from Agent A's work. Intelligence cannot compound; it fragments.

More fundamentally, when interacting decisions are evaluated against incompatible representations of reality, the system admits no single coherent execution history. Each decision may be locally justifiable relative to what it observed, yet there exists no unified history in which the joint outcomes of these decisions can be explained. Once such decisions commit irreversible effects, later reconciliation cannot retroactively restore coherence.

\begin{definition}[Decision Coherence Law]
For agents taking irreversible actions whose effects interact, operating constructively requires that interacting decisions be evaluated against a coherent representation of reality at the moment they are made.
\end{definition}

This law is foundational in scope, but its necessity follows directly from the structure of interacting decisions. We do not reduce it to more primitive laws; rather, we establish its necessity by showing that without it, the operational meaning of collective behavior is undefined rather than merely suboptimal.

\subsection{Operational Requirements}

The Decision Coherence Law defines a structural condition for constructive collective operation. This section explains how that condition must be interpreted in practice, in systems where reality evolves continuously and agent decisions have effects that extend beyond system boundaries.

From the law, three categories of requirements follow:

\subsubsection{Semantic Operations (Scope)}

The transactional scope of the system must extend beyond raw records to include semantic meaning and transformations that influence decisions. Similarity relations, inferred state, intent, and other derived interpretations are not external annotations; they participate in the same coherence regime as base data.

\textbf{Requirement.} The system must support semantic transformation and semantic retrieval as native capabilities over live state.

\textbf{Invariant.} Semantic interpretation that participates in decision evaluation is represented within the system's state and governed by the same coherence guarantees as other decision-relevant data.

\subsubsection{Consistency Properties (Safety)}

Without transactional consistency, agents may observe intermediate or partially applied changes. Such states never existed as valid system configurations, and decisions evaluated against them are undefined.

\textbf{Requirement.} All state transitions must satisfy standard transactional consistency guarantees (linearizability, serializability, snapshot isolation, or read committed).

\textbf{Invariant.} At any moment, there exists exactly one authoritative representation of state against which decisions are evaluated.

\subsubsection{Operational Envelopes (Liveness)}

Because reality evolves continuously and agent actions take effect outside transactional control, systems must define the temporal and concurrency bounds under which observations remain admissible for decision making.

\textbf{Temporal Envelope ($\Delta$).} Agents operating constructively over shared context must adopt a system-wide temporal bound $\Delta$ that limits the admissible gap between retrieval time and decision time for any decision with shared effects.

For every observation used in a decision: $\text{decision\_time} - \text{retrieval\_time} < \Delta$.

\textbf{Concurrency Envelope ($C$).} Agents operating constructively over shared context must adopt a system-wide concurrency level $C$, such that transactional consistency and the temporal bound $\Delta$ are maintained for all decisions under sustained concurrent operation at that level.

\subsection{Necessity and Sufficiency}

Decision Coherence requires the joint satisfaction of semantic operations, transactional consistency, and operational envelopes. These constraints address different failure modes and are independent but inseparable: each is necessary, and none is sufficient on its own.

\begin{itemize}[leftmargin=*]
\item \textbf{Transactional Consistency without Temporal Envelope} produces coherent but stale reality. Agents agree on what was true, not what is true now.
\item \textbf{Transactional Consistency without Concurrency Envelope} produces correct transactions that cannot complete within decision windows under load.
\item \textbf{Temporal + Concurrency without Transactional Consistency} produces fast but incoherent views. Agents see recent state but divergent versions of it.
\item \textbf{All three without Semantic Operations} produces identical data with incompatible interpretations.
\end{itemize}

\section{The Necessity of the Context Lake}

\subsection{Composition Limits}

\begin{lemma}[Visibility Gating Necessity]\label{lem:visibility}
Consider multiple systems, each determining when state becomes visible according to internal system policy. An external coordinator cannot guarantee cross-system atomic visibility of decision-relevant effects while keeping the participating systems within their native system classes unless every participating system exposes a visibility-gating capability to:
\begin{enumerate}[leftmargin=*,noitemsep]
\item durably accept decision-relevant effects while keeping them invisible, and
\item later make those effects visible or abort them based on an explicit external decision.
\end{enumerate}
\end{lemma}

\begin{proof}[Proof sketch]

Atomic visibility must be enforced at write-time or read-time. 

\textit{Write-time:} In a system lacking visibility-gating, once an effect is durably applied it becomes visible according to internal system policy. In any execution without instantaneous, globally synchronized visibility, there exists an interval in which an effect has become visible in one system while the corresponding effect has not yet become visible in another, violating atomic visibility.

\textit{Read-time:} Another approach to enforcing atomicity is to read each system “as of” a shared cut at query time. However, each system determines when state becomes visible according to its own internal policy, and these visibility boundaries are not shared or comparable across systems. Any shared cut must therefore be defined in an external coordinate system (e.g., event time). Once coherence is defined relative to an event-time cut, mutation operations are no longer native---all changes must be represented as immutable event-time facts. Consequently, the systems degenerate into append-only logs with state semantics implemented externally.
\end{proof}

\begin{definition}[Independently Advancing System]\label{def:independent}
A system is independently advancing if it determines when durable state becomes visible according to internal system policy and does not expose any mechanism by which visibility of those effects can be made contingent on an external boundary.
\end{definition}

This describes the default behavior of modern infrastructure: search engines advance via indexing cycles; caches via time-to-live (TTL) and invalidation; analytical systems via batch refresh; replicas via replication lag; vector databases via asynchronous indexing.

\begin{theorem}[Composition Impossibility]\label{theo:impossibility}
Decision coherence under continuous mutation cannot be achieved by composing independently advancing systems (Definition~\ref{def:independent}) while preserving the independently advancing system class, unless a single system already enforces decision coherence as a non-bypassable property.
\end{theorem}

\begin{proof}

Under continuous mutation, decision coherence requires that all decision-relevant effects become visible atomically across the systems that participate in the decision, within the decision window.

By Lemma~\ref{lem:visibility}, cross-system atomic visibility of such effects cannot be guaranteed unless every participating system exposes a visibility-gating capability. By Definition~\ref{def:independent}, an independently advancing system exposes no such capability.

It follows that composing independently advancing systems while preserving the independently advancing system class cannot satisfy the necessary condition for decision coherence under continuous mutation.

The only exception is when a single system already enforces decision coherence as a non-bypassable property, so that decisions do not require cross-system atomic visibility guarantees from a composition.
\end{proof}

\begin{corollary}[Authority Localization]
Under continuous mutation, decision coherence can be enforced only within a single system boundary. It cannot be synthesized by composing multiple systems without degeneration of their system classes.
\end{corollary}

\subsection{Survey of Existing System Classes}

We now examine why each existing system class individually fails to satisfy Decision Coherence.

\textbf{Relational Databases} provide transactional consistency, but not decision-time temporal or concurrency envelopes over derived state or ad-hoc retrieval patterns. They lack native support for semantic operations.

\textbf{NoSQL/Document Stores} provide scale and flexibility, but offer weaker consistency guarantees, restricted retrieval patterns, and lack native semantic operations.

\textbf{Search Engines} provide flexible full-text retrieval but lack transactional consistency. Index refresh introduces unavoidable staleness.

\textbf{Data Warehouses} provide shared analytical snapshots but lack live ingestion, live transformation, and low-latency retrieval under concurrency.

\textbf{Data Lakes and Lakehouses} provide shared storage at rest but lack live ingestion, live transformation, and low-latency retrieval under concurrency.

\textbf{Stream Processing Systems} provide live ingestion and transformation, but are not designed to serve context—raw or derived—for low-latency retrieval.

\textbf{Vector Databases} provide semantic similarity retrieval but support only this single access pattern and lack general-purpose retrieval capabilities.

Each system class fails to satisfy Decision Coherence on its own, even though every individual requirement of Decision Coherence exists somewhere today. By the Composition Impossibility Theorem~\ref{theo:impossibility}, no composition of existing system classes satisfies Decision Coherence.

\section{Context Lake: The System Class Defined by Decision Coherence}

\subsection{Definition}

\begin{definition}[Context Lake]
A Context Lake is a system that enforces Decision Coherence by:
\begin{enumerate}[leftmargin=*,noitemsep]
\item maintaining a single authoritative logical representation of reality for decisions,
\item ensuring timely visibility of all decision-relevant mutations,
\item preserving guarantees under sustained concurrency,
\item executing semantic interpretation and retrieval natively and authoritatively.
\end{enumerate}
\end{definition}

A system qualifies as a Context Lake if and only if it enforces Decision Coherence: transactional consistency over all decision-relevant state, maintained within declared temporal and concurrency envelopes, with semantic operations executed as native capabilities.

\textbf{Relationship to existing classes.} Context Lake requirements form a strict superset of existing system class requirements. A Context Lake can subsume the roles of a database, warehouse, search engine, and vector store, while additionally enforcing the constraints required for Decision Coherence.

A system may initially fall outside the Context Lake class. As the system evolves to enforce the invariants required for Decision Coherence, it becomes a Context Lake by definition. This classification depends on enforced invariants rather than prior system class, lineage, or architectural resemblance.

\subsection{Context Preparation and Context Retrieval}

The purpose of a Context Lake is to provide the right context at the moment a decision is made. Context participates in a Context Lake in two distinct phases:

\textbf{Context Preparation} organizes experience into memory that may later support decisions. This phase transforms raw observations into representations that may be relevant for future decisions through interpretation, consolidation, and structuring.

\textbf{Context Retrieval} is when prepared or on-demand context is retrieved to evaluate and decide action. In this phase, the system retrieves the specific subset of context required to evaluate a decision predicate. Retrieved context determines the outcome of the decision and constrains which actions are admissible.

\subsection{Context Lake in the Decision Loop}

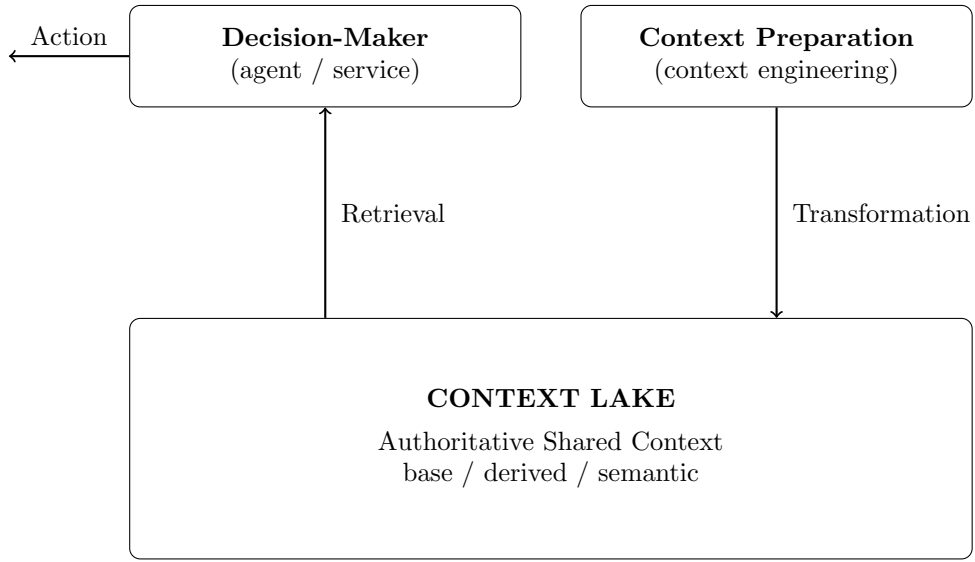
\begin{figure}[t]
\centering
\begin{tikzpicture}[
  font=\small,
  box/.style={
    draw,
    rounded corners,
    minimum width=5.2cm,
    minimum height=1.35cm,
    align=center
  },
  lake/.style={
    draw,
    rounded corners,
    minimum width=11.2cm,
    minimum height=3.2cm,
    align=center
  },
  arrow/.style={->, thick}
]

\node[lake] (lake) at (0,-0.8) {
  \textbf{CONTEXT LAKE}\\[4pt]
  Authoritative Shared Context\\
  \small base / derived / semantic
};

\node[box, anchor=south west] (decision)
  at ([yshift=2.8cm] lake.north west) {
  \textbf{Decision-Maker}\\
  (agent / service)
};

\node[box, anchor=south east] (prep)
  at ([yshift=2.8cm] lake.north east) {
  \textbf{Context Preparation}\\
  (context engineering)
};

\coordinate (lakeTopUnderDecision) at (lake.north -| decision.south);
\coordinate (lakeTopUnderPrep) at (lake.north -| prep.south);

\draw[arrow] (lakeTopUnderDecision) -- (decision.south)
  node[midway, right, xshift=2pt]{Retrieval};

\draw[arrow] (prep.south) -- (lakeTopUnderPrep)
  node[midway, right, xshift=2pt]{Transformation};

\draw[arrow] (decision.west) -- ++(-1.6,0)
  node[midway, above]{Action};

\end{tikzpicture}
\caption{Context Lake as the authoritative shared substrate for decision coherence.}\label{fig:context-lake-architecture}
\end{figure}

A Context Lake sits between experience and action (Figure~\ref{fig:context-lake-architecture}). It provides a shared substrate through which
decision-makers obtain the context required to evaluate and make decisions under real-world
constraints. The Context Lake does not prescribe how decisions are made. Instead, it defines how
experience is organized into decision-relevant context and how that context is retrieved at the
boundary. Decisions are evaluated against context at retrieval time.

Decision-makers interact with a Context Lake in exactly two ways: by retrieving context to evaluate
actions, and by contributing experience through their actions.

An action may take effect in the external world (e.g., presenting a recommendation) or contribute
new experience by updating shared context. All decision logic remains external to the system class.

\section{Context Engineering: Making Shared Memory Coherent}

A Context Lake provides infrastructure. Context Engineering concerns context preparation: the organization of memory prior to context retrieval.

\subsection{Memory Layers: The Canonical Structure}

Context Engineering organizes memory into three distinct layers, each with different mutability contracts, write authorities, and roles in supporting decision making.

We introduce these layers---episodic, semantic, and state memory---and show how
Context Engineering governs the transformations between them, defining how recorded
experience informs shared interpretation and operative conditions.

\subsubsection{Episodic Memory (Append-Only)}

\textbf{Invariant:} Episodic memory preserves observed experience as it occurred. Past experience is never revised to change what was observed.

\textbf{Mutability:} Immutable. Episodes are not revised in place. Corrections are recorded as new episodes.

\textbf{Contents:} Events, logs, messages, traces, and other raw observations.

\textbf{Role:} Foundation of memory. The authoritative record of observation from which meaning and state are derived.

\subsubsection{Semantic Memory (Governed)}

\textbf{Invariant:} Semantic memory represents shared interpretation of experience. Interpretations may evolve, but they are maintained as shared context.

\textbf{Governed By:} Written only by Context Engineering through explicit, versioned semantic transformations.

\textbf{Mutability:} Mutable by design. Interpretations may be revised as evidence, models, or definitions improve, without altering underlying episodic memory.

\textbf{Contents:} Interpreted signals—entity resolutions, sentiment classifications, extracted concepts, causal links, relationships.

\textbf{Role:} Interprets experience. Bridges raw observation and interpretation by answering ``What does this mean?'' without prescribing action.

\subsubsection{State Memory (Mutable)}

\textbf{Invariant:} State memory represents the current set of operative conditions used to evaluate decisions.

\textbf{Governed By:} Governed by Context Engineering. Agents, applications, and policies may update state by executing permitted, well-defined transitions under shared consistency and validation rules.

\textbf{Mutability:} Mutable by design. State evolves as conditions change. Updates are applied with transactional guarantees under concurrency.

\textbf{Contents:} Authoritative present-tense facts such as current status flags, thresholds, quotas, counters, and active constraints.

\textbf{Role:} Authoritative for action evaluation. Agents query state memory when evaluating and executing decisions, applying updates transactionally against it.

\subsection{Why These Layers Are Distinct}

The three memory layers represent fundamentally different kinds of information with different lifecycle requirements. Episodic memory records what was observed---changing it would falsify history. Semantic memory records what observations mean---it must be mutable because understanding evolves. State memory records what is operative now---it must be mutable because conditions change, and authoritative because agents require a single source of conditions when evaluating actions.

Collapsing observation, interpretation, and state into a single undifferentiated layer leads to rewriting or versioning history (as interpretations evolve), and to blurring intermediate analysis with decision-ready truth. These failures are structural consequences of conflating concerns.

\section{Agent Decision Admissibility Conditions}

Context Lake provides transactional consistency and bounded temporal and concurrency envelopes over decision-relevant context.
These guarantees are necessary but not sufficient on their own for Decision Coherence: they define the conditions under which coherent decisions are possible.
Whether coherence is realized in practice further depends on how agents evaluate and act on that context.

This section defines closure conditions that constrain agent behavior and decision logic. These conditions are not enforced automatically by Context Lake. Rather, they specify the admissibility requirements that agent decisions must satisfy to be considered coherent.

\subsection{Elimination of Private Decision Premises}

\textbf{Condition:} No decision with shared effects may be justified by decision-relevant context that is private to a single agent.

When decision-relevant context is agent-local, there is no guarantee that interacting decisions are evaluated against the same premises. For a decision with shared effects to be meaningful, all context that influences its admissibility must be part of shared context.

\subsection{Elimination of Mixed Causal Cuts}

\textbf{Condition:} A single decision may not be evaluated against multiple causal cuts of context.

If a decision is evaluated using context drawn from multiple causal cuts, the decision is evaluated against a world that never existed. Such a decision has no coherent model under which it could be justified.

\subsection{Elimination of Deferred Action}

\textbf{Condition:} An agent may not take an action with shared effects based on context retrieved outside its admissible decision window.

Context retrieved from a Context Lake is authoritative at the moment it is observed. However, if an agent delays action beyond the temporal envelope under which that context remains admissible, the decision is no longer evaluated against the reality in which its effects take place.

\subsection{Elimination of Implicit Semantics}

\textbf{Condition:} Meaning that influences a decision must be explicit and shared.

If interpretation exists only within application logic, model prompts, or agent-local behavior, identical observations may yield incompatible interpretations across agents. Decisions may then diverge despite operating over the same underlying data.

\section{Related Work}

We relate Context Lake to prior work across several areas of data management and distributed systems.

\textbf{Transaction processing.} The ACID properties~\cite{haerder1983principles} define correctness for transactional systems. Our work extends transactional guarantees to decision-time semantics, where the evaluation of a decision and its effects may not occur within a single transaction boundary.

\textbf{Distributed consistency.} The CAP theorem~\cite{brewer2000towards,gilbert2002brewer} establishes fundamental trade-offs in distributed systems. Context Lake requires strong consistency (CP systems) but adds temporal and concurrency envelopes as additional requirements.

\textbf{Eventual consistency and CRDTs.} Systems like Dynamo~\cite{decandia2007dynamo} and CRDTs~\cite{shapiro2011comprehensive} provide eventual consistency with convergence guarantees. Decision Coherence requires stronger guarantees---decisions must observe consistent state at decision time, not eventually.

\textbf{CALM and monotonicity.} The CALM theorem~\cite{hellerstein2010declarative,alvaro2011consistency} shows that monotonic programs can be executed without coordination. Decision Coherence addresses a complementary problem: when decisions are non-monotonic and interact, what must the underlying system guarantee?

\textbf{Stream processing and dataflow.} Systems like Naiad~\cite{murray2013naiad}, Apache Flink~\cite{carbone2015flink}, Spark Streaming~\cite{zaharia2013discretized}, and Dataflow~\cite{akidau2015dataflow} provide live computation over evolving data. 
These systems compute derived context but do not serve it for low-latency retrieval at high concurrency.

\textbf{Data lakehouses.} Delta Lake~\cite{armbrust2020delta} and similar systems unify lakes and warehouses for analytics workloads. Context Lake extends this concept to operational agent decision making with decision coherence guarantees.

\textbf{HTAP systems.} HTAP systems aim to support transactional and analytical workloads within a single platform. However, HTAP lacks formal semantics: it is introduced as a workload characterization rather than a system class with defined correctness properties.

\textbf{HSAP systems.} Hybrid Serving and Analytical Processing (HSAP) systems~\cite{jiang2020hsap} focus on integrating low-latency serving with analytical workloads. HSAP addresses performance and data access at the serving boundary, but does not define correctness semantics for concurrent, irreversible decisions over shared decision-time context.

\textbf{Multi-agent systems.} Classical multi-agent systems research~\cite{wooldridge2009introduction} addresses coordination, communication, and collective behavior. Our work provides a data systems foundation for multi-agent correctness at scale.

\section{Discussion and Future Work}

\subsection{Implementation Challenges}

Implementing a Context Lake presents significant engineering challenges:

\textbf{Achieving temporal envelopes.} Maintaining bounded staleness for derived context under continuous mutation and high concurrency requires incremental maintenance strategies that preserve consistency.
For semantic operations, this challenge is amplified when evaluation requires external computation, such as calls to large language models. In such cases, LLM-mediated interpretations must remain within the admissible temporal envelope.

\textbf{Supporting diverse retrieval patterns.} A Context Lake must support all decision-relevant retrieval patterns effectively, including point lookups, range scans, ad-hoc filtering, similarity search, and semantic retrieval, without violating transactional guarantees or temporal envelopes.

\textbf{Scalability.} Context Lakes must scale to support many concurrent agents operating over large volumes of raw and derived context. Standard database scaling techniques may not suffice.

\textbf{Workload isolation.} Context preparation and context retrieval operate concurrently over the same shared Context Lake. Isolating these workloads without delaying admissible retrieval is challenging: preparation must not block decision-time reads, while retrieval must not observe partially updated or out-of-envelope derived context.

\subsection{Open Questions}

We conclude by highlighting several open questions raised by the Context Lake model.

\textbf{Relaxed consistency models.} Can weaker consistency models (e.g., causal consistency) suffice for certain agent interaction patterns? Under what conditions?

\textbf{Independence of decisions.} When decisions do not interact, they need not observe the same causal cut. Can systems exploit independence to improve performance while preserving Decision Coherence for interacting decisions?

\textbf{Semantic accuracy vs.\ performance trade-offs.} How do accuracy-performance trade-offs for semantic operations manifest under decision-time constraints?

\textbf{Admissibility enforcement.} How should systems detect and prevent actions that violate admissible decision windows, causal cuts, or shared semantic premises at runtime?

\subsection{Evaluation Agenda}

Validating the Context Lake framework requires:

\begin{enumerate}[leftmargin=*]
\item \textbf{Prototype implementation} demonstrating that semantic operations, consistency guarantees, and operational envelopes can be satisfied simultaneously.
\item \textbf{Performance characterization} establishing achievable temporal and concurrency envelopes under realistic workloads.
\item \textbf{Case studies} showing real-world scenarios where Context Lakes prevent failures that arise in loosely composed or multi-system architectures.
\item \textbf{Comparative evaluation} quantifying benefits over alternative architectures under comparable workloads.
\end{enumerate}

\section{Conclusion}

This paper identifies Decision Coherence as a fundamental requirement for correctness when autonomous AI agents make concurrent, irreversible decisions over shared resources. We show that no existing system class satisfies this requirement and prove through the Composition Impossibility Theorem that independently advancing systems cannot be composed to provide decision coherence while preserving their native system classes.

From this impossibility result, Context Lake emerges as a necessary system class, defined by three requirements: native semantic operations, transactional consistency over all decision-relevant state, and operational envelopes that bound staleness and degradation under load.

We formalize the architectural invariants, enforcement boundaries, and admissibility conditions required for correctness in collective agent systems. While significant implementation challenges remain, the theoretical foundation establishes that Context Lakes are not an optimization or architectural preference, but a structural necessity for agent systems operating constructively at scale.

As AI agents become increasingly prevalent as the primary consumers of data systems, the principles established here will shape the next generation of data infrastructure—systems designed not for human analysis, but for collective machine intelligence.

\section*{Acknowledgments}
I thank colleagues for thoughtful discussions and feedback during the development of this work.

\bibliographystyle{plain}

\end{document}